\documentclass[a4paper,USenglish,cleveref,numberwithinsect]{lipics-v2021}

\usepackage[utf8]{inputenc}
\usepackage{mathtools,stmaryrd,anyfontsize}
\usepackage{microtype}
\usepackage{booktabs}
\usepackage{tikz,float}

\hideLIPIcs
\nolinenumbers
\numberwithin{equation}{section}
\newcommand{\embed}{\lhook\joinrel\xrightarrow{}}
\newcommand{\just}{\lhook\joinrel\xrightarrow{\text{just}}}

\newcommand{\N}{\mathbb N}
\newcommand{\Z}{\mathbb Z}
\newcommand{\CompSubSeq}{\mathsf{CompSubSeq}}
\newcommand{\SubSeq}{\mathsf{SubSeq}}
\newcommand{\PSPACE}{\mathsf{PSPACE}}
\newcommand{\NC}{\mathsf{NC}}
\newcommand{\TC}{\mathsf{TC}}
\newcommand{\DLOGTIME}{\mathsf{DLOGTIME}}
\newcommand{\val}{\operatorname{\mathsf{val}}}
\newcommand{\bval}{\operatorname{\mathsf{bval}}}
\newcommand{\Eq}{\operatorname{\mathsf{Eq}}}
\newcommand{\Save}{\operatorname{\mathsf{Save}}}
\newcommand{\Ret}{\operatorname{\mathsf{Return}}}
\newcommand{\Add}{\operatorname{\mathsf{Add}}}
\newcommand{\Sub}{\operatorname{\mathsf{Sub}}}
\newcommand{\enc}{\mathsf{c}}
\newcommand{\blk}{\mathsf{b}}
\newcommand{\evalprog}[1]{\llbracket #1\rrbracket}
\newcommand{\ind}[1]{\mathbf{1}[#1]}

\title{Compressed Subsequence Checking is PSPACE-complete}
\titlerunning{Compressed Subsequence Checking is PSPACE-complete}
\author{Markus Lohrey}{University of Siegen, Germany}{lohrey@eti.uni-siegen.de}{https://orcid.org/0000-0002-4680-7198}{}
\authorrunning{M. Lohrey}
\Copyright{Markus Lohrey}


\begin{document}
\maketitle

\begin{abstract}
It is shown that the (scattered) subsequence problem for two words represented by
straight-line programs is $\PSPACE$-complete, even over a binary
alphabet. The lower bound is obtained by a polynomial-time reduction
from quantified subset sum. If the two words are given explicitly, i.e., uncompressed,
then the subsequence problem is hard for uniform $\NC^1$ and in deterministic logspace.
\end{abstract}

\newpage
\setcounter{page}{1}

\section{Introduction}

A \emph{straight-line program} (SLP) is a context-free grammar with
exactly one production for each variable and an acyclic dependency
relation. It represents the unique word derived from its start variable.
By reusing variables, an SLP can copy previously defined factors in a word and thereby
achieve exponential compression in the best case. Due to their generality,
SLPs provide a suitable model for studying algorithms that work
directly on compressed words. The aim is to solve problems in time
polynomial in the compressed input size, without first expanding the
represented words; see the survey~\cite{Lohrey2012}.

Several fundamental results underpin the importance of the SLP model.
The first polynomial-time algorithms for testing whether two SLPs produce the same word
were developed independently by Hirshfeld et al.~\cite{HirshfeldJM94}, Mehlhorn et al.~\cite{MehlhornSU94}, and
Plandowski~\cite{Plandowski1994}. Rytter~\cite{Rytter2003}
and Charikar et al.~\cite{CharikarEtAl2005} established close connections
between Lempel--Ziv compression and SLPs. In particular, a
Lempel--Ziv parse can be converted into an SLP with only a logarithmic
increase in size. These connections allow algorithms for SLP-compressed
words to be applied to other compressed representations as well.

A central problem in this area is \emph{compressed pattern matching}:
given an explicitly represented pattern $p$ and an SLP for a text $t$,
one asks whether $p$ occurs as a contiguous factor of $t$.
Ganardi and Gawrychowski~\cite{GanardiGawrychowski2022} proved that this
problem can be solved in time $\mathcal{O}(g_t+|p|)$, where $g_t$ denotes the
size of the SLP for the word $t$ in the following. Thus, even when the text is exponentially
long, the running time is linear in the total size of the input.

In \emph{fully compressed pattern matching}, both the pattern and the
text are represented by SLPs. Karpi\'nski, Rytter, and
Shinohara~\cite{KarpinskiRytterShinohara1995} gave the first
polynomial-time algorithm for this problem. 
Miyazaki, Shinohara, and Takeda~\cite{MiyazakiShinoharaTakeda1997}
obtained an algorithm with running time $\mathcal{O}(g_p^2g_t^2)$,
and Lifshits~\cite{Lifshits2007} improved this bound to
$\mathcal{O}(g_pg_t^2)$. Je\.{z}~\cite{Jez2015} subsequently used
recompression to obtain an algorithm with running time
$\mathcal{O}((g_p+g_t)\log |p|)$. The algorithms of Lifshits and
Je\.{z} also compute a compact representation of all occurrences, whose number may be
exponential in the compressed input size. The running-time bounds stated
here use the word-RAM model over a fixed alphabet, with machine words
large enough to store positions in the uncompressed words.

Allowing gaps between consecutive letters of an occurrence leads to the
\emph{fully compressed subsequence problem} studied in this paper.
The input for this problem consists of two
straight-line programs that represent words $p$ and $t$, and the question
is whether $p$ is a subsequence of $t$. Equivalently, $p$ must be obtained
by deleting letters at arbitrary positions from $t$. 
Lifshits and Lohrey proved that the fully compressed subsequence problem is
$\Theta_2^p$-hard and belongs to
$\PSPACE$~\cite[Theorem~4 and Proposition~1]{LifshitsLohrey2006}; later the lower bound
was strengthened to
$\mathsf{PP}$-hardness~\cite[Theorem~13]{Lohrey2011}.
In this paper, we prove the following result. We denote the fully
compressed subsequence problem over an alphabet of size $n$ by $\CompSubSeq_n$.

\begin{theorem}\label{thm:main}
For every $n \geq 2$, $\CompSubSeq_n$ is $\PSPACE$-complete under
polynomial-time many-one reductions. 
\end{theorem}
\Cref{thm:main} settles the open problem on the precise complexity of
fully compressed subsequence matching raised in~\cite{GanardiSaglamZetzsche2024,LifshitsLohrey2006,Lohrey2012}.

We prove the lower bound by a reduction from quantified subset sum~\cite{Travers2006}.
First, we show how to represent a monotone map on a finite integer
interval by a binary word. Reading this word in the greedy embedding
algorithm applies the map to an offset in the pattern. We then construct
a recursive sequence of such maps that evaluates the quantified
subset sum instance. The two branches of a quantifier use the same
sequence of maps, but start with different offsets. Finally, we compute in polynomial time straight-line programs for the binary words
representing all maps used in the construction.

The focus of this work is on compressed words. Nevertheless, in view of \Cref{thm:main} it is natural to ask for the 
precise complexity of the uncompressed subsequence problem, where the input consists of two explicitly represented (uncompressed)
words $p$ and $t$, and the question is whether $p$ is a subsequence of $t$. We denote this problem with $\SubSeq_n$,
where $n$ denotes again the alphabet size. It is easy to see that this problem can be solved in linear time as well as in deterministic logspace.
In Section~\ref{sec:uncompressed} we show the following result:

\begin{theorem}\label{thm:mainNC1}
$\SubSeq_2$ is hard for $\NC^1$ under $\TC^0$ many-one reductions.\footnote{All circuit complexity classes in this paper are $\DLOGTIME$-uniform.
Definitions regarding circuit complexity
can be found in \Cref{sec:complexity}.}
\end{theorem}
The precise complexity of $\SubSeq_2$ remains open.

\subparagraph{Further related work.} Besides the fully compressed subsequence problem, also its semi-compressed version has been studied \cite{CegielskiEtAl2006,Tiskin2009,Tiskin2011,YamamotoEtAl2011}.
In this problem, the input consists of an explicitly given pattern $p$ and an SLP for a text $t$ and it is asked whether $p$ is a subsequence of $t$.
This problem can be solved in time $\mathcal{O}(|p| \cdot g_t)$ \cite{YamamotoEtAl2011}.

\section{Preliminaries}\label{sec:preliminaries}

Let $\N=\{0,1,2,\ldots\}$. For integers $a,b$, write
$[a,b]=\{z\in\Z:a\leq z\leq b\}$; in particular,
$[a,b]=\varnothing$ if $a>b$. All intervals in the following are
integer intervals. 
For functions $f : A \to B$ and $g : B \to C$ we define their composition $f \circ g : A \to C$
by $(f \circ g)(a) = g(f(a))$ for $a \in A$.
For a statement $\mathcal S$ that can be either true or false, let $\ind{\mathcal S}$
be $1$ if $\mathcal S$ is true and $0$ otherwise. All logarithms are
to base two.

\subsection{Complexity} \label{sec:complexity}

We assume that the reader is familiar with the basic concepts of complexity theory. We will mainly 
deal with the class $\PSPACE$ (polynomial space). For the last section, we need a little bit of circuit
complexity. The reader can skip these definitions in case she/he is only interested in Theorem~\ref{thm:main}.

The complexity class $\NC^1$ contains all languages $L \subseteq \{0,1\}^*$ for which there is a family $(\mathcal{B}_n)_{n \in \mathbb{N}}$ of Boolean circuits
with the following properties:
\begin{itemize}
\item $\mathcal{B}_n$ is a Boolean circuit with $n$ input gates. All other gates compute one of the functions
 $\neg$ (unary negation), $\wedge$ (binary conjunction) and $\vee$ (binary disjunction).
\item $\mathcal{B}_n$ has size $\mathsf{poly}(n)$  and depth $\mathcal{O}(\log n)$.
\item For every $w = a_1 a_2 \cdots a_n$ ($a_i \in \{0,1\}$), $w \in L$ if and only if $\mathcal{B}_n$
outputs $1$ when the $i^{th}$ input gate is set to $a_i$ for all $i \in [1,n]$.
\end{itemize}
An important subclass of $\NC^1$ is $\TC^0$, which is defined similarly to $\NC^1$ with the following modifications:
\begin{itemize}
\item Conjunction and disjunction gates may have arbitrary fan-in and in addition, the circuits $\mathcal{B}_n$ may use
majority gates of arbitrary fan-in.  
The $m$-ary majority function returns $1$ if and only if more than $m/2$ many input bits are $1$.
\item The depth of every circuit $\mathcal{B}_n$ is bounded by a fixed constant.
\end{itemize}
We only use the $\DLOGTIME$-uniform variants of  $\TC^0$ and $\NC^1$.
For $\TC^0$, $\DLOGTIME$-uniformity means that
for a given tuple $(1^n, u, v)$, where $1^n$ is the unary encoding of $n \in \mathbb{N}$ and
 $u$ and $v$ are binary encoded gates of the $n$-th circuit $\mathcal{B}_n$, one can 
\begin{enumerate}[(i)]
\item compute the label of gate $u$ in time
$\mathcal{O}(\log n)$ and
\item \label{point-ii} check in  time $\mathcal{O}(\log n)$ whether $u$ is an input gate for $v$.
\end{enumerate}
Note that since the number of gates of $\mathcal{B}_n$ is polynomially bounded in $n$, the gates of $\mathcal{B}_n$ can be encoded with 
$\mathcal{O}(\log n)$ bits. 

The definition of $\DLOGTIME$-uniform $\NC^1$ is similar, but instead of \eqref{point-ii}
one requires that for given $1^n$, $u$, $v$ as above and an address string $\alpha \in \{0,1\}^*$ with
$|\alpha| \leq \log_2 |\mathcal{B}_n|$ one can check in time $\mathcal{O}(\log n)$ whether $v$ can be reached from $u$ by descending along
the directions in $\alpha$ (where $0$ stands for moving to the left input of the current gate and a $1$ 
stands for moving to the right input of the current gate)
 \cite{BIS90,Ruz81,Vollmer99}. In the following we always refer to the $\DLOGTIME$-uniform variants when we talk
 about $\NC^1$ and $\TC^0$.
It is known that $\NC^1$ coincides with $\mathsf{ALOGTIME}$
(logarithmic time on an alternating random access Turing machine).

The definitions of the above circuit complexity classes 
can be easily extended to functions $f\colon \{0,1\}^* \to \{0,1\}^*$
by using circuits with several output gates together with an appropriate padding so that outputs for equal-length inputs have the same length. 
This yields the notion of a $\TC^0$ many-one reduction.

Typical problems in $\TC^0$ are the computation of the integer quotient of binary encoded integers, and the sum and product of an arbitrary number of binary encoded integers \cite{HeAlBa02}. The canonical $\TC^0$-complete language is
$\mathsf{Majority} = \{ x \in \{0,1\}^* : x \text{ has more $1$'s than $0$'s} \}$. 

\subsection{Words} \label{sec:words}

Let $\Sigma$ be a finite alphabet. The empty word is denoted by $\varepsilon$.
For $w\in \Sigma^*$, let $|w|$ be its length and let $w[i]$ be its $i$-th
letter, where $i\in[1,|w|]$. A word $u=a_1\cdots a_m$ is a
\emph{subsequence} of $v$ (or $u$ \emph{embeds} into $v$), written $u\embed v$, if there are positions
$1\leq i_1<\cdots<i_m\leq|v|$ with $v[i_j]=a_j$ for all
$j\in[1,m]$. One can check $u \embed v$ using a simple greedy embedding algorithm that maps every
symbol $a_i$ to the earliest possible position in $v$. More precisely, $a_1$ is mapped to the first
position $j_1$ with $a_1 = v[j_1]$ (if it exists), then $a_2$ is mapped to the first position $j_2>j_1$ with $a_2 = v[j_2]$ (if it exists), and so on.
This algorithm works in linear time and logarithmic space. For the latter, notice that it suffices to store
a pointer into $u$ as well as a pointer into $v$.

We say that $u \neq \varepsilon$ \emph{just embeds} into $v \neq \varepsilon$, $u \just v$ in symbols,
if $u \embed v = wa$ for $a \in \Sigma$ but $u \not\embed w$.
This means that the above greedy algorithm maps the last symbol of $u$ to the last symbol of $v$.
Note that $u_1 \just v_1$ and $u_2 \just v_2$ imply $u_1u_2 \just v_1v_2$.

\subsection{Straight-line programs} \label{sec:slp}

A \emph{straight-line program} (SLP) over $\Sigma$ is a context-free grammar
$\mathcal G=(V,\Sigma,S,P)$ with the following properties:
\begin{itemize}
\item For every variable $X \in V$ there is exactly one production of the form $X \to \alpha$.
\item $\mathcal G$ is acyclic, i.e., there is no non-empty derivation of the form $X \Rightarrow_{\mathcal G}^+ \alpha X \beta$.
\end{itemize}
This ensures that every variable $X$ generates exactly one word, denoted by
$\val_{\mathcal G}(X)$, and $\val(\mathcal G)=\val_{\mathcal G}(S)$.
We omit the subscript $\mathcal{G}$ when the SLP is clear from the context. The size
of an SLP is the sum of the lengths of its right-hand sides.

An SLP $\mathcal G=(V,\Sigma,S,P)$ can be represented by a directed acyclic graph (DAG for short), whose vertex set is $V \cup \Sigma$ (as usual for 
context-free grammars, we assume that $V \cap \Sigma = \emptyset$). Moreover, if $X \to \alpha_0 \alpha_1 \cdots \alpha_n$ is the unique
production for the variable $X$, where $\alpha_i \in V \cup \Sigma$ for $i \in [0,n]$, then there is an $i$-labelled edge from $X$ to $\alpha_i$
in the DAG, which is written as the triple $(X,i,\alpha_i)$. The \emph{depth} of $\mathcal G$ is the length of a longest directed path (measured in number of edges) 
in the DAG of $\mathcal G$.

We may assume that all productions are of the form $X\to a$ or
$X\to YZ$ for $X,Y,Z \in V$ and $a \in \Sigma$.\footnote{This does not allow producing the empty word, which is not a problem for our purpose.} This corresponds to the well-known Chomsky normal form.
 In this normal form, the size and the number of
productions differ by at most a constant factor. We also allow powers
$X^k$, with $k\in\N$ given in binary notation, in right-hand sides.
Such an expression can be easily replaced by $\mathcal{O}(\log(k+1))$ ordinary
productions using iterated squaring; see also
\cite[Section~2]{LifshitsLohrey2006}.

The problem $\CompSubSeq_n$ is defined as follows, where $\Sigma$ is an alphabet of size $n$:
\begin{quote}
\textsc{Input:} Two SLPs $\mathcal P,\mathcal T$ over $\Sigma$.\\
\textsc{Question:} Does $\val(\mathcal P)\embed\val(\mathcal T)$ hold?
\end{quote}
It is easy to see that this problem belongs to 
$\PSPACE$ \cite[Proposition~1]{LifshitsLohrey2006}. One can run the above-mentioned
greedy embedding algorithm on the words $\val(\mathcal P)$ and $\val(\mathcal T)$ without
computing $\val(\mathcal P)$ and $\val(\mathcal T)$ explicitly. The pointers into $\val(\mathcal P)$ and $\val(\mathcal T)$, respectively,
need polynomial space.

\subsection{Quantified subset sum}\label{sec:source-problem}

We use the quantified subset sum problem as defined by 
Travers~\cite[p.~214]{Travers2006}:
\begin{quote}
\textsc{Input:} Binary-encoded $w_1, \ldots, w_n, t \in \mathbb{N}$ and quantifiers $Q_1, \ldots, Q_n \in \{\exists, \forall \}$ \\
\textsc{Question:} Does $\ind{\Phi}=1$ hold for the formula
\begin{equation}\label{eq:source}
 \Phi=Q_1x_1\in\{0,1\}\ \cdots\ Q_nx_n\in\{0,1\}:
  \sum_{i=1}^n x_iw_i=t.
\end{equation}
\end{quote}
By a reduction from quantified Boolean satisfiability, it is shown in \cite{Travers2006} that
quantified subset sum is $\PSPACE$-complete.

\section{Proof of \texorpdfstring{\Cref{thm:main}}{Theorem 1.1}}

In this section we prove \Cref{thm:main}. By \cite[Proposition~1]{LifshitsLohrey2006} it suffices to 
show $\PSPACE$-hardness. For this we reduce quantified subset sum to $\CompSubSeq_2$.
Let us fix the formula $\Phi$ from \eqref{eq:source} with 
$w_1, \ldots, w_n, t \in \mathbb{N}$ and $Q_1, \ldots, Q_n \in \{\exists, \forall \}$.
 Let $S=\sum_{i=1}^n w_i$.
We can assume that $t \in [0,S]$ and $n\geq1$; other cases are trivial.

Let us give a short overview of the proof, before we go into the details. 
The formula $\Phi$ in \eqref{eq:source} can be unfolded into an exponentially large Boolean formula tree $T_\Phi$ by replacing universal (resp., existential)
quantifiers by conjunctions (resp., disjunctions). The leaves of the tree $T_\Phi$ are identities $s=t$, where $s$ is a partial
sum of some of the weights $w_i$. Every such identity can be true or false. 
For instance, for $\Phi = \forall x_1\in\{0,1\}\ \exists x_2\in\{0,1\}:
 x_1+x_2=1$ the tree is shown in Figure~\ref{fig:tree}.

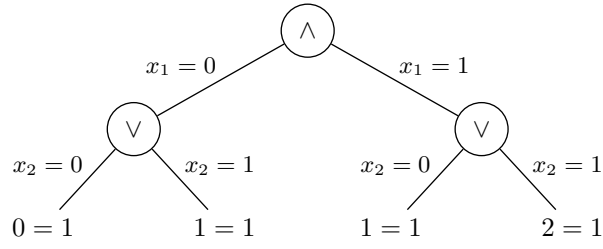
\begin{figure}[t]
\centering
\begin{tikzpicture}[
  scale = 1,
  every node/.style={font=\normalsize},
  gate/.style={draw,circle,minimum size=7mm,inner sep=0pt},
  edge label/.style={font=\small,inner sep=1.5pt},
  line width=0.5pt
]
  \node[gate] (root) at (0,0) {$\wedge$};
  \node[gate] (left) at (-2.3,-1.3) {$\vee$};
  \node[gate] (right) at (2.3,-1.3) {$\vee$};
  \node (leaf00) at (-3.5,-2.6) {$0=1$};
  \node (leaf01) at (-1.1,-2.6) {$1=1$};
  \node (leaf10) at (1.1,-2.6) {$1=1$};
  \node (leaf11) at (3.5,-2.6) {$2=1$};

  \draw (root) -- node[edge label,above left] {$x_1=0$} (left);
  \draw (root) -- node[edge label,above right] {$x_1=1$} (right);
  \draw (left) -- node[edge label,above left] {$x_2=0$} (leaf00);
  \draw (left) -- node[edge label,above right] {$x_2=1$} (leaf01);
  \draw (right) -- node[edge label,above left] {$x_2=0$} (leaf10);
  \draw (right) -- node[edge label,above right] {$x_2=1$} (leaf11);
\end{tikzpicture}
\caption{The Boolean formula tree $T_\Phi$ for
$\Phi=\forall x_1\in\{0,1\}\ \exists x_2\in\{0,1\}:x_1+x_2=1$.}
\label{fig:tree}
\end{figure}

The tree $T_\Phi$ can be evaluated using a depth-first-left-to-right
traversal and a stack of height $n$ (the number of quantifiers in \eqref{eq:source}). This evaluation process will be simulated by a sequence
of functions $f_1, f_2, \ldots, f_\ell$ with $f_i : [0,M] \to [0,M]$ (we will call these functions `instructions').
Every $f_i$ is monotone and has some additional properties that are defined in Section~\ref{sec:simulation}. The sequence
$f_1, f_2, \ldots, f_\ell$ will be constructed in Sections~\ref{sec:maps}--\ref{sec:program}.
The numbers $\ell$ and $M$ are exponential in the size of the formula $\Phi$.
The precise correspondence between the formula \eqref{eq:source} and the sequence of functions
$f_1, f_2, \ldots, f_\ell$ is that $f_\ell(f_{\ell-1}(\cdots f_1(1)\cdots)) \geq 2$ if and only if $\Phi$ is true.

Due to the special properties of the functions $f_i$, we can construct from the sequence $f_1, f_2, \ldots, f_\ell$
two exponentially long words $u, v \in \{0,1\}^*$ such that
$f_\ell(f_{\ell-1}(\cdots f_1(1)\cdots)) \geq 2$ if and only if $u \embed v$. This will be done in \Cref{sec:simulation}.
All that remains to be shown is that $u$ and $v$ can be obtained by straight-line programs that can be constructed
in polynomial time from the formula $\Phi$. For $u$, this is easy; it is $(0^M 1^M)^\ell 00$. For $v$ we argue 
as follows: First of all, the sequence of instructions $f_1 f_2 \cdots f_\ell$ can be produced by a small SLP due to the regular
structure of the tree $T_\Phi$. Moreover, $v$ is obtained by replacing in the sequence $f_1 f_2 \cdots f_\ell$ every $f_i$ by
a code $\blk_M(f_i) \in \{0,1\}^*$ (in addition we  have to add a leading $0$ to $v$). Moreover, for each of the instructions that we use, the code $\blk_M(f_i)$ can be also
produced by a small SLP; see \Cref{sec:codes}. This concludes our proof sketch.

\subsection{Monotone maps and subsequence embeddings}\label{sec:simulation}

A function $f:[a,b]\to[a,b]$ on a nonempty integer interval $[a,b]$
 is \emph{monotone} if for all $x,y \in [a,b]$, $x \leq y$ implies $f(x)\leq f(y)$.
The composition of two monotone functions
$f,g:[a,b]\to[a,b]$ is again monotone. 
Note that being monotone does not imply $f(x)\geq x$ for all $x$ in the domain of $f$.

For an integer $M\geq1$, a function $f:[0,M]\to[0,M]$ is called
\emph{admissible} (on $[0,M]$) if it is monotone,
$f(0)=0$, $f(M)=M$ and $f([1,M-1]) \subseteq [1,M-1]$.
For such a function, we define the words
 \begin{eqnarray}\label{eq:code}
 \enc_M(f) &=& 1^{f(M)-f(M-1)}0 1^{f(M-1)-f(M-2)}0 \cdots 1^{f(1)-f(0)}0 \quad\text{and}\\
  \blk_M(f) &=& \enc_M(f)(01)^{M-1} 0.
\end{eqnarray}
Each $\enc_M(f)$ contains exactly
$M$ zeros and $M$ ones, so $|\blk_M(f)|=4M-1$. Moreover, note that $\enc_M(f)$ starts with $1$ (since $f(M-1) < M = f(M)$) and ends
with $10$  (since $f(1) > 0 = f(0)$).
We will use this encoding only for maps $f$ whose words $\enc_M(f)$ have polynomial-size SLPs. 

\begin{lemma}\label[lemma]{lem:simulation-1}
Assume that $u\, 0^z \just v$ for words $u,v$ and $z \in [1,M-1]$. 
Let $f$ be an admissible function on $[0,M]$.
Then $u\, 0^M 1^M 0^{f(z)} \just v \,\blk_M(f)$ holds.
\end{lemma}

\begin{proof}
The power $0^{M-z}$ just embeds into the prefix
\[ 
1^{f(M)-f(M-1)}0 1^{f(M-1)-f(M-2)}0 \cdots 1^{f(z+1)-f(z)}0
\]
of $\blk_M(f)$.
Hence, we have
\[
u \,0^{M} \just v 1^{f(M)-f(M-1)}0 1^{f(M-1)-f(M-2)}0 \cdots 1^{f(z+1)-f(z)}0.
\]
The remaining suffix $1^{f(z)-f(z-1)}0 1^{f(z-1)-f(z-2)}0 \cdots 1^{f(1)-f(0)}0$ of 
$\enc_M(f)$ contains exactly $f(z)$ many $1$s. Hence 
\[ u\, 0^M 1^{f(z)} \just v 1^{f(M)-f(M-1)}0 1^{f(M-1)-f(M-2)}0 \cdots 1^{f(1)-f(0)}. \]
The word on the right-hand side is $v\, \enc_M(f)$ without the $0$ at the end.
Note that $f(z) \in [1,M-1]$ since $f$ is admissible and $z \in [1,M-1]$. Hence $M-f(z) > 0$ and we obtain
\[ u \, 0^M 1^{M} = u\, 0^M 1^{f(z)} 1^{M-f(z)} \just v \, \enc_M(f) (01)^{M-f(z)}. \]
Finally, this yields
$u \, 0^M 1^{M} 0^{f(z)} \just v\, \enc_M(f) (01)^{M-f(z)}  (01)^{f(z)-1} 0 = v\, \blk_M(f)$.
\end{proof}
A straightforward inductive application of \Cref{lem:simulation-1} yields:
\begin{lemma}\label[lemma]{lem:simulation}
Let $\ell\geq1$, let $f_1, \ldots, f_\ell$ be admissible functions on $[0,M]$ and $z \in [1,M-1]$.
Define $f = f_1 \circ f_2 \circ \cdots \circ f_\ell$.
If $u \, 0^z \just v$ for words $u,v$ then 
\[u(0^M 1^M)^\ell 0^{f(z)} \just v \,\blk_M(f_1) \blk_M(f_2) \cdots \blk_M(f_\ell).\]
\end{lemma}
We apply \Cref{lem:simulation} for $u=\varepsilon$, $v = 0$ and $z=1$, which yields
$u 0^z \just v$ and hence
\[(0^M 1^M)^\ell 0^{f(z)} \just 0 \,\blk_M(f_1) \blk_M(f_2) \cdots \blk_M(f_\ell).\]
This implies that
\begin{equation}\label{eq:final-threshold}
(0^M 1^M)^\ell 00 \embed 0 \,\blk_M(f_1) \blk_M(f_2) \cdots \blk_M(f_\ell) \quad \Longleftrightarrow \quad f(z) \geq 2.
\end{equation}

\subsection{Simulation of quantified subset sum by admissible functions}\label{sec:maps}

Recall that $n \ge 1$ is the number of variables in the quantified subset sum formula \eqref{eq:source} and $S = \sum_{i=1}^n w_i$.
We define for all $i \in [1,n]$:
\begin{equation}\label{eq:parameters}
 A=2^n,\quad B=4A = 2^{n+2} ,\quad 
 C_i=\frac{B}{2^i} = 2^{n+2-i} ,\quad D_i=2C_i, \quad M=B(S+1).
\end{equation}
Thus $C_i \ge 4$, $4\!\mid\! C_i$, $M \ge B \ge 8$,
$D_i\!\mid\! B$ and $B \!\mid\! M$, and
\begin{equation}\label{eq:weights}
 \sum_{i=1}^n C_i= 4 \sum_{i=1}^n 2^{n-i} = 4 \sum_{i=0}^{n-1} 2^{i} = 
 4 (2^{n}-1) = B-4.
\end{equation}
 All the numbers in \eqref{eq:parameters} have $\mathcal{O}(n+\log(S+1))$ bits.

\begin{figure}[t]
\centering
\begin{tikzpicture}[
  scale = 0.8,
  line width=0.5pt,
  dot/.style={circle,fill,inner sep=0pt,minimum size=4pt},
  edge label/.style={font=\small,inner sep=0pt},
  leaf label/.style={below,inner sep=0.5pt}
]
  \node[dot] (n1) at (0,0) {};
  \node[dot] (n2) at (1.3,-1.5) {};
  \node[dot] (n3) at (2.6,-3) {};
  \node[dot] (n4) at (1.3,-4.5) {};
  \node[dot] (n5) at (2.6,-6) {};
  \node[leaf label] (v) at (3.9,-7.5) {$v$};

  \coordinate (b1) at (-1.3,-1.5);
  \coordinate (b2) at (0,-3);
  \coordinate (other3) at (3.9,-4.5);
  \coordinate (b4) at (0,-6);
  \coordinate (b5) at (1.3,-7.5);

  \draw (n1) -- node[edge label,above left] {$x_1=0$} (b1)
    node[leaf label] {$b_1$};
  \draw (n1) -- node[edge label,above right] {$x_1=1$} (n2);

  \draw (n2) -- node[edge label,above left] {$x_2=0$} (b2)
    node[leaf label] {$b_2$};
  \draw (n2) -- node[edge label,above right] {$x_2=1$} (n3);

  \draw (n3) -- node[edge label,above left] {$x_3=0$} (n4);
  \draw (n3) -- node[edge label,above right] {$x_3=1$} (other3);

  \draw (n4) -- node[edge label,above left] {$x_4=0$} (b4)
    node[leaf label] {$b_4$};
  \draw (n4) -- node[edge label,above right] {$x_4=1$} (n5);

  \draw (n5) -- node[edge label,above left] {$x_5=0$} (b5)
    node[leaf label] {$b_5$};
  \draw (n5) -- node[edge label,above right] {$x_5=1$} (v);
  
\end{tikzpicture}
\caption{A path from the root of $T_\Phi$ to a node $v$. Every $b_i$ is the truth value ($0$ or $1$) of the corresponding subtree.}
\label{fig:state}
\end{figure}
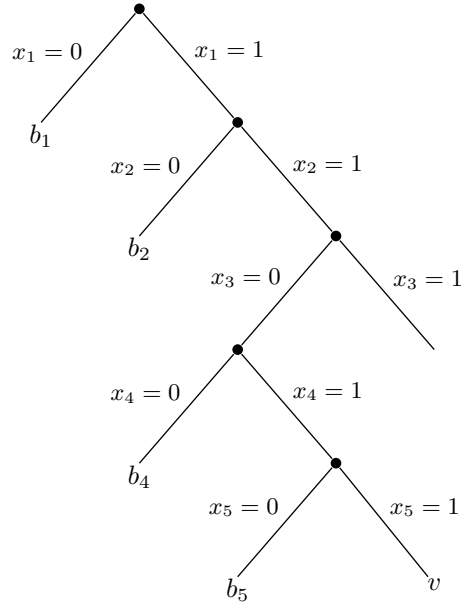

We will mainly deal with admissible functions $f :[0,M]\to[0,M]$ in the following.
Such functions will be also called \emph{instructions}.
During the execution of the instruction sequence that we are going to construct, the current number $z \in [0,M]$ will be called the \emph{state}.
It  will always have the form
\begin{equation}\label{eq:state}
 z=Bs+\sum_{j=1}^n C_jb_j+r,
 \qquad 0\leq s\leq S,\quad b_j\in\{0,1\},\quad r\in\{1,2\}.
\end{equation}
Let us explain the meaning of the numbers $s$, $b_1$, \ldots, $b_n$, $r$.
Recall the Boolean formula tree $T_\Phi$ that is obtained from unfolding $\Phi$.
The instruction sequence that we are going to construct evaluates $T_\Phi$
by doing a depth-first-left-to-right traversal of $T_\Phi$. Consider the situation in Figure~\ref{fig:state}, which shows a path from the root
of $T_\Phi$ to a node $v$ that is not necessarily a leaf of $T_\Phi$.
 The bits $b_1, b_2, b_4, b_5 \in \{0,1\}$ 
are the truth values to which the corresponding subtrees of $T_\Phi$ evaluate. These bits 
are needed for evaluating the nodes on the path from the root to $v$
and therefore
have to be stored
in the state $z$ when the traversal is in the subtree rooted at $v$.
At the time instant when the tree traversal enters $v$ from its parent node,  the state is 
$(w_1+ w_2 + w_4 + w_5) \cdot B + (C_1 b_1 + C_2 b_2 + C_4 b_4 + C_5 b_5) + 1$. In particular, in \eqref{eq:state}
all bits $b_i$ with $i \notin \{1,2,4,5\}$ are zero. Note that $s := w_1 + w_2 + w_4 + w_5$ is the current partial sum; $w_3$ is not included
since the path from the root to $v$ takes the edge with $x_3=0$.
When the traversal goes from $v$ back to its parent node the
state $z$ will be $(w_1+ w_2 + w_4 + w_5) \cdot B + (C_1 b_1 + C_2 b_2 + C_4 b_4 + C_5 b_5) + 1 + \beta$, where 
$\beta \in \{0,1\}$ is the truth value to which the subtree rooted in $v$ evaluates. 
By~\eqref{eq:weights}, every state in~\eqref{eq:state} satisfies
\begin{equation}\label{eq:state-bound}
 1\leq z\leq BS+B-2=M-2.
\end{equation}
The numbers $s, b_1, \ldots, b_n, r$ in \eqref{eq:state}
are uniquely determined by the state $z$: the lower part 
$\sum_{j=1}^n C_jb_j+r$ is strictly less
than $B$, each $C_j$ is a different power of two and at least four, and
$r\in\{1,2\}$ occupies only the two least significant binary positions.

Recall that an instruction is an admissible map on $[0,M]$.
We often construct instructions from admissible mappings $h : [0,D] \to [0,D]$ for some $D\!\mid\! M$.
For this, we use the following lifting lemma:

\begin{lemma}\label[lemma]{lem:lifting}
Let $D \ge 1$ be a divisor of $M$.
For every $q \in [0, M/D-1]$, let
$h_q:[0,D]\to[0,D]$ be admissible and
define $f : [0,M] \to [0,M]$ by
$f(qD+r)=qD+h_q(r)$ for $r \in [0,D-1]$ and $f(M)=M$.
Then $f$ is admissible.
\end{lemma}
\begin{proof}
Monotonicity within an interval $[qD, qD+D-1]$ follows from monotonicity of $h_q$.
At the upper interval boundary we have
\[
 f(qD+D-1)\leq qD+D-1<(q+1)D=f((q+1)D).
\]
We have $f(0) = h_0(0) = 0$. By definition, $f(M) = M$ holds.
For $z \in [1,M-1]$, write $z=qD+r$ with $0\leq r \le D-1$. If $r>0$, then
$f(z)=qD+h_q(r)\geq1$. If $r=0$, then $q\geq1$ and $f(z)=qD>0$. The condition
$f(z) \le M-1$ follows similarly.
\end{proof}
In the following we define the instructions that are used for the evaluation of the Boolean formula tree $T_\Phi$.

\subsubsection{Equality tests} \label{sec-eq}
We start with the instructions for the leaves of $T_\Phi$, where the truth of a statement $s=t$ has to be checked.
Here, $s$ is a partial sum of the $w_i$.

For $b\in\{0,1\}$, define an admissible map $h_b : [0,4] \to [0,4]$ by
\begin{equation}\label{eq:local-equality}
 h_b(0)=0,\qquad h_b(1)=h_b(2)=h_b(3)=1+b,\qquad h_b(4)=4.
\end{equation}
For $z=4q+r<M$, $0\leq r<4$, let
\begin{equation}\label{eq:equality}
 \Eq(4q+r)=4q+h_{\ind{\lfloor4q/B\rfloor=t}}(r),
 \qquad \Eq(M)=M.
\end{equation}
Here, $\lfloor4q/B\rfloor=t$ is seen as a Boolean statement whose truth value $\ind{\lfloor4q/B\rfloor=t}$ depends only on $q$.

 \Cref{lem:lifting}, with $D=4$, 
proves that $\Eq$ is admissible.
On a state of the form \eqref{eq:state} with $r=1$, $\Eq$ yields
\begin{equation}\label{eq:equality-action}
 \Eq\biggl(Bs+\sum_{j=1}^n C_jb_j+1\biggr)
 =Bs+\sum_{j=1}^nC_jb_j+1+\ind{s=t}.
\end{equation}
To verify this equation, put $a=\sum_j C_jb_j$. Then $4\!\mid\! (Bs+a)$ and
$0\leq a\leq B-4$ by \eqref{eq:weights}. Hence, the unique interval $[4q,4q+3]$
containing the input $Bs+a+1$ has left endpoint
$4q=Bs+a$ and $\lfloor4q/B\rfloor=s$. Note that the partial sum $s$ and all stored bits $b_j$
are not changed by the instruction $\Eq$.

\begin{remark}
Although $s\mapsto\ind{s=t}$ need not be monotone, the map $\Eq$ is
monotone. The reason is that it retains the partial sum and the saved
bits, and changes only the least significant component (the $r$ in \eqref{eq:state}).
\end{remark}

\subsubsection{Addition and subtraction}

While traversing the tree $T_\Phi$ we have to add and subtract weights $w_i$ to the state component $s$ in  \eqref{eq:state}.
This will be done by the following admissible functions $\Add_i$ and $\Sub_i$ ($i \in [1,n]$), where $z \in [1,M-1]$:
\begin{align}
 \Add_i(0)&=0,&  \Add_i(M)&=M,&  \Add_i(z)&=\min\{z+Bw_i,M-1\},
 \label{eq:add}\\
 \Sub_i(0)&=0,& \Sub_i(M)&=M,& \Sub_i(z)&=\max\{z-Bw_i,1\}.
 \label{eq:subtract}
\end{align}
Note that $0 \le Bw_i\leq BS=M-B\leq M-4$.
The application of $\min$ and $\max$  in \eqref{eq:add} and \eqref{eq:subtract}
is needed to obtain admissible functions. When applying $\Add_i$ and $\Sub_i$
in the instruction sequence for evaluating $T_\Phi$, we will always have
$\Add_i(z) = z+Bw_i$ and $\Sub_i(z) = z-Bw_i$.

\subsubsection{Saving and combining branch results} \label{sec-save-ret}

Recall the saved bits $b_i$ in the  generic state $z$ from  \eqref{eq:state}.
These bits will be set by instructions $\Save_i$ and $\Ret_i$.

Fix $i$ and recall the definition of $C_i$ and $D_i=2C_i$ from \eqref{eq:parameters}.
Note that $D_i\!\mid\! M$. Define
\begin{eqnarray}
\label{eq:sigma}
 \sigma_i(y) &=&
 \begin{cases}
 0 & \text{if } y=0,\\
 1 & \text{if } y=1,\\
 C_i+1 & \text{if } 2\leq y<D_i,\\
 D_i & \text{if } y=D_i.
 \end{cases} \\
\label{eq:theta}
 \theta_i &=&
 \begin{cases}
 2 & \text{if } Q_i=\exists,\\
 C_i+2 & \text{if } Q_i=\forall,
 \end{cases}
 \qquad
 \tau_i(y)=
 \begin{cases}
 0 & \text{if } y=0,\\
 1 & \text{if } 1\leq y<\theta_i,\\
 2 &  \text{if } \theta_i\leq y<D_i,\\
 D_i & \text{if } y=D_i.
 \end{cases}
\end{eqnarray}
Since $C_i+1 < D_i$ and $D_i > 2$, 
the maps $\sigma_i$ and $\tau_i$ are both admissible. Lift them periodically to $[0,M]$, where
$q \in [0, M/D_i-1]$ and $y \in [0,D_i-1]$:
\begin{equation}\label{eq:lift}
 \Save_i(qD_i+y)=qD_i+\sigma_i(y), \quad
 \Ret_i(qD_i+y)=qD_i+\tau_i(y),
\end{equation}
and set $\Save_i(M) = \Ret_i(M)=M$.
These liftings are admissible by \Cref{lem:lifting}.

For every $\alpha\in\{0,D_i,2D_i,\ldots,M-D_i\}$ and  $b \in \{0,1\}$,
\begin{equation}\label{eq:save-action}
 \Save_i(\alpha+1+b)= \alpha + \sigma_i(1+b) = 
 \alpha+C_ib+1.
\end{equation}
Thus $\Save_i$ stores the bit $b$ in the $i$-th saved bit $b_i$ from \eqref{eq:state} and resets
the result (i.e., sets $r=1$ in \eqref{eq:state}). 
Here, we should remark that whenever we apply $\Save_i$ in our final instruction sequence, the state
$z$ will have the above form $\alpha+1+b$ with $\alpha\in\{0,D_i,2D_i,\ldots,M-D_i\}$ and  $b \in \{0,1\}$.
The bit $b$ in \eqref{eq:save-action} should be seen as the truth value of the left child of the current tree node 
during the traversal of $T_\Phi$. When the truth value $b' \in \{0,1\}$ of the second child node is also determined, the 
current state will have the form $\alpha + C_ib+1 + b'$ with the same  $\alpha\in\{0,D_i,2D_i,\ldots,M-D_i\}$ as in \eqref{eq:save-action}.
Note that $C_ib+1 + b' \le C_i+2 < 2C_i = D_i$ (since $C_i > 2$).
The following table shows the value of $\tau_i(C_ib+1 + b')$ for the four possible combinations of $(b,b')$:
\begin{center}
\begin{tabular}{ccccc}
\toprule
$(b,b')$ & $(0,0)$ & $(0,1)$ & $(1,0)$ & $(1,1)$\\
\midrule
$C_ib+1 + b'$ & $1$ & $2$ & $C_i+1$ & $C_i+2$\\
$Q_i = \exists$ & $1$ & $2$ & $2$ & $2$\\
$Q_i = \forall$ & $1$ & $1$ & $1$ & $2$\\
\bottomrule
\end{tabular}
\end{center}
Consequently, for every $\alpha\in\{0,D_i,2D_i,\ldots,M-D_i\}$ and  $b,b' \in \{0,1\}$ we have
\begin{equation}\label{eq:return-action}
 \Ret_i(\alpha+C_ib+1+b')=
 \begin{cases}
 \alpha+1+(b\vee b')& \text{if } Q_i=\exists,\\
 \alpha+1+(b\wedge b')& \text{if } Q_i=\forall.
 \end{cases}
\end{equation}
In particular, the $i$-th saved bit is cleared again.

\subsection{Recursive evaluation}\label{sec:program}

An instruction word (i.e., a sequence of admissible functions on $[0,M]$)
is executed from left to right. For such a word
$F=f_1\cdots f_\ell$, write
$\evalprog{F}=f_1\circ\cdots\circ f_\ell : [0,M] \to [0,M]$. Here, the empty instruction word
yields the identity mapping. Note that
$\evalprog{FG}=\evalprog{F}\circ\evalprog{G}$.
Each instruction below is one of the admissible maps from Sections~\ref{sec-eq}--\ref{sec-save-ret}.
Define inductively the instruction words
\begin{equation}\label{eq:program}
 E_{n+1}=\Eq,\qquad
 E_i=E_{i+1}\,\Save_i\,\Add_i\,E_{i+1}\,\Sub_i\,\Ret_i
 \quad(1\leq i\leq n).
\end{equation}
These equations form an SLP of size $\mathcal{O}(n)$ over an instruction alphabet
of size $4n+1$. 
For a partial sum $s \in [0,S]$, define the Boolean value
\begin{equation}\label{eq:suffix-truth}
 F_i(s)= \ind{Q_ix_i\in\{0,1\}\cdots Q_nx_n\in\{0,1\} : s+\sum_{j=i}^n x_jw_j=t}.
  \end{equation}
In particular, $F_{n+1}(s)=\ind{s=t}$.

\begin{lemma}\label{lem:evaluation}
For $1\leq i\leq n+1$, $0\leq s\leq S-\sum_{j=i}^n w_j$, and $b_1,\ldots,b_{i-1} \in \{0,1\}$ we have
\begin{equation}\label{eq:evaluation-invariant}
\evalprog{E_i}\biggl( Bs+\sum_{j=1}^{i-1} C_jb_j +1 \biggr) = Bs+\sum_{j=1}^{i-1} C_jb_j +1+F_i(s).
\end{equation}
\end{lemma}
\begin{proof}
The proof is by downward induction on $i$. At $i=n+1$,
\eqref{eq:evaluation-invariant} is precisely
\eqref{eq:equality-action}.
For the induction step, set
$\beta_0=F_{i+1}(s)$ and $\beta_1=F_{i+1}(s+w_i)$.
We use the abbreviation
\[
 \alpha=Bs+\sum_{j=1}^{i-1} C_jb_j.
\]
Since $B$ and every $C_j$ with $j<i$ are divisible by $D_i = 2C_i$, so is
$\alpha$. Moreover,
\[
 \sum_{j=1}^{i-1} C_j =  \sum_{j=1}^{i-1} 2^{n+2-j} =
 \sum_{j=1}^{i-1} 2^{n+2-i+j} 
 = 2^{n+3-i} (2^{i-1}-1) = 2^{n+2} - 2^{n+3-i} = B-D_i
\]
and hence $0\leq\alpha\leq BS+B-D_i=M-D_i$.

We have to show that $\evalprog{E_i}(\alpha+1) = \alpha+1+F_i(s)$.
For this we will apply the instructions 
$\evalprog{E_{i+1}}$, $\evalprog{\Save_i}$, $\evalprog{\Add_i}$, $\evalprog{E_{i+1}}$,
$\evalprog{\Sub_i}$ and $\evalprog{\Ret_i}$ in this order to $\alpha+1$.
For the first call $\evalprog{E_{i+1}}(\alpha+1)$, the induction hypothesis is applicable because
\[
 s\leq S-w_i-\sum_{j=i+1}^n w_j
    \leq S-\sum_{j=i+1}^n w_j.
\]
The additional saved bit $b_i$ is initially zero. We obtain $\evalprog{E_{i+1}}(\alpha+1) = \alpha+1+\beta_0$.
Since $\alpha \leq M-D_i$ is a multiple of $D_i$, \eqref{eq:save-action} 
yields
$\Save_i(\alpha+1+\beta_0) = \alpha+C_i\beta_0+1$.
Note that 
\[\alpha+B w_i +C_i\beta_0+1 = B(s+w_i) + \sum_{j=1}^{i-1}C_jb_j + C_i\beta_0+1 \leq M-2 < M-1 \]
 by \eqref{eq:state-bound}.
Hence, by  \eqref{eq:add} we get $\Add_i(\alpha+C_i\beta_0+1) = \alpha+B w_i +C_i\beta_0+1$.

For the second recursive call $\evalprog{E_{i+1}}(\alpha+B w_i +C_i\beta_0+1)$ 
the induction hypothesis applies again: We have 
$\alpha+B w_i +C_i\beta_0+1 =  B(s+w_i) + \sum_{j<i}C_jb_j + C_i\beta_0+1$ and
\[
 s+w_i\leq S-\sum_{j=i+1}^n w_j.
\]
We obtain $\evalprog{E_{i+1}}(\alpha+B w_i +C_i\beta_0+1) = \alpha+Bw_i+C_i\beta_0+1+\beta_1$.
By  \eqref{eq:subtract} and $(\alpha+Bw_i+C_i\beta_0+1+\beta_1)-Bw_i \geq 1$ we next get
\[
\Sub_i(\alpha+Bw_i+C_i\beta_0+1+\beta_1) = \alpha+C_i\beta_0+1+\beta_1.
\]
Finally, applying $\Ret_i$ yields by \eqref{eq:return-action}
\[ 
\Ret_i(\alpha+C_i\beta_0+1+\beta_1) = \begin{cases}
 \alpha+1+(\beta_0\vee\beta_1)& \text{if } Q_i=\exists,\\
 \alpha+1+(\beta_0\wedge\beta_1)& \text{if }  Q_i=\forall.
 \end{cases}
 \]
 Here is the entire execution sequence again:
\begin{eqnarray*}
 \alpha+1
 &\xmapsto{E_{i+1}} & \alpha+1+\beta_0\\
 &\xmapsto{\Save_i} & \alpha+C_i\beta_0+1\\
 &\xmapsto{\Add_i} & \alpha+Bw_i+C_i\beta_0+1\\
 &\xmapsto{E_{i+1}} & \alpha+Bw_i+C_i\beta_0+1+\beta_1\\
 &\xmapsto{\Sub_i} & \alpha+C_i\beta_0+1+\beta_1\\
 &\xmapsto{\Ret_i} &
 \begin{cases}
 \alpha+1+(\beta_0\vee\beta_1),&Q_i=\exists,\\
 \alpha+1+(\beta_0\wedge\beta_1),&Q_i=\forall.
 \end{cases}
\end{eqnarray*}
By the semantics of $Q_i$, the final state is
$\alpha+1+F_i(s)$. 
Note that every intermediate number that arises in the above calculation belongs to $[1,M-1]$.
\end{proof}
For $s=0$ and $i=1$, \Cref{lem:evaluation} yields
\begin{equation}\label{eq:root}
 \evalprog{E_1}(1)=1+F_1(0)=1+\ind{\Phi}.
\end{equation}
If $\ell_i=|E_i|$, then $\ell_{n+1}=1$ and
$\ell_i=2\ell_{i+1}+4$. Therefore
\begin{equation}\label{eq:program-length}
 \ell:=|E_1|=5\cdot2^n-4.
\end{equation}

\subsection{Straight-line programs for the instruction codes}\label{sec:codes}

It remains to show that for every instruction map $f$ from Sections~\ref{sec-eq}--\ref{sec-save-ret} one can
compute in polynomial time an SLP for
the code $\blk_M(f)$. Since the suffix $(01)^{M-1} 0$ of every code $\blk_M(f)$ can be easily 
produced by a small SLP, it suffices to consider the word $\enc_M(f)$.
For each of these words we will construct an SLP
that uses a constant number of concatenations and powers with $\mathcal{O}(\log M)$-bit exponents.
As explained in \Cref{sec:slp}, each power can be replaced by
ordinary SLP-productions. 

\subsubsection{Lifting preserves compression}
Suppose $D\geq2$ divides $M$ and, on consecutive intervals of length $D$, the map $f$ has
the form
\[
 f(qD+r)=qD+h_q(r)  \quad (r \in [0,D-1]),
\]
where every $h_q$ is admissible on $[0,D]$ and $f(M)=M$. The successive differences in
\eqref{eq:code} give
\begin{equation}\label{eq:code-lift}
 \enc_M(f)=\enc_D(h_{M/D-1})\cdots\enc_D(h_1)\enc_D(h_0).
\end{equation}
To see this, let $j=qD+r$ with $1\leq r\leq D$. Then
\[
 f(j)-f(j-1)=h_q(r)-h_q(r-1).
\]
This identity also holds for $r=D$, since 
\begin{eqnarray*}
f(qD+D)-f(qD+D-1) & = & f((q+1)D)-f(qD+D-1) \\
& = & (q+1)D - qD - h_q(D-1) \\
& =  & D - h_q(D-1) \\
& = & h_q(D) - h_q(D-1).
\end{eqnarray*}
 This proves~\eqref{eq:code-lift}.
If all $h_q$ are equal to $h$, then \eqref{eq:code-lift} specializes to
\begin{equation}\label{eq:uniform-code}
 \enc_M(f)=\enc_D(h)^{M/D}.
\end{equation}

\subsubsection{Equality tests}
For the maps in~\eqref{eq:local-equality}, put
\begin{equation}\label{eq:small-equality-codes}
 L_0=\enc_4(h_0)=11100010,\qquad
 L_1=\enc_4(h_1)=11000110.
\end{equation}
Consider now the function 
$\Eq$ from \eqref{eq:equality}, where $q$ ranges over $[0,M/4-1] = [0,A(S+1)-1]$.
Consider the Boolean statement $\lfloor4q/B\rfloor=t$ whose truth value only depends on $q$. 
Since $B = 4A$ it is equivalent to $\lfloor q/A\rfloor=t$. 
For $q \in \{0, \ldots, tA-1\}$ ($tA$ values) $\lfloor q/A\rfloor=t$ is false, for 
$q \in \{tA,\ldots, tA + A-1\}$ ($A$ values) $\lfloor q/A\rfloor=t$ is true, and for 
$q \in \{(t+1)A,\ldots, A(S+1)-1\}$ ($(S-t)A$ values) $\lfloor q/A\rfloor=t$ is false.
As a consequence, we obtain
\begin{equation}\label{eq:equality-code}
 \enc_M(\Eq)=L_0^{(S-t)A}\,L_1^A\,L_0^{tA}.
\end{equation}
Here the order of the $L_b$-powers is due to the descending order of the arguments for $f$ in \eqref{eq:code}. 

\subsubsection{Addition and subtraction}
With $a = Bw_i \in [0,M-4]$, direct calculation of the successive differences
in~\eqref{eq:add} and \eqref{eq:subtract} yields
\begin{align}
 \enc_M(\Add_i)=10\,0^a\,(10)^{M-a-2}\,1^{a+1}0,
 \label{eq:add-code}\\
 \enc_M(\Sub_i)=1^{a+1}0\,(10)^{M-a-2}\,0^a\,10.
 \label{eq:sub-code}
\end{align}
The successive differences of $\Add_i$ are:
$\Add_i(1)-\Add_i(0)=a+1$,
$\Add_i(z)-\Add_i(z-1)=1$ for $2 \le z \le M-a-1$, 
$\Add_i(z)-\Add_i(z-1)=0$ for $M-a \le z \le M-1$
and $\Add_i(M)-\Add_i(M-1)=1$.
Listing these differences in reverse
order gives~\eqref{eq:add-code}. 
For $\Sub_i$, the differences are 
$\Sub_i(1)-\Sub_i(0)=1$,
$\Sub_i(z)-\Sub_i(z-1)=0$ for $2 \le z \le a+1$, 
$\Sub_i(z)-\Sub_i(z-1)=1$ for $a+2 \le z \le M-1$
and $\Sub_i(M)-\Sub_i(M-1)=a+1$.

\subsubsection{Saving and returning}
With $C=C_i$, $D=D_i$, and $\theta=\theta_i$, the only nonzero
successive differences of $\sigma_i$ are
$\sigma_i(1)-\sigma_i(0)=1$, $\sigma_i(2)-\sigma_i(1)=C$, and $\sigma_i(D)-\sigma_i(D-1)=C-1$.
All other differences are zero. Therefore 
\[\enc_D(\sigma_i)=1^{C-1}0^{D-2}1^C0\,10.\]
For $\tau_i$, the nonzero differences are
$\tau_i(1)-\tau_i(0)=1$, $\tau_i(\theta)-\tau_i(\theta-1)=1$, and $\tau_i(D)-\tau_i(D-1)=D-2$,
giving 
\[\enc_D(\tau_i)=1^{D-2}0^{D-\theta}1\,0^{\theta-1}10.\]
Note that since $C\geq4$ and
$2\leq\theta\leq C+2<D$, all exponents are nonnegative.
Finally, for $\Save_i$ and $\Ret_i$ we obtain
$\enc_M(\Save_i)=\enc_D(\sigma_i)^{M/D}$ and
$\enc_M(\Ret_i)=\enc_D(\tau_i)^{M/D}$.

\subsubsection{Computing straight-line programs for the codes}

There are $4n+1$ instruction maps $f$. The code $\blk_M(f)$ for each instruction map $f$
is obtained by a constant
number of concatenations and exponentiations $u^K$ where $K$ has $\mathcal{O}(\log M)$ bits.
Hence, one can compute in polynomial time SLPs for all codes $\blk_M(f)$. The total number
of productions is $\mathcal{O}(n \log M) = \mathcal{O}(n^2 + n \log S)$.

\subsection{The final reduction}\label{sec:reduction}

We now construct the final SLPs and complete the proof of
\Cref{thm:main}.

\begin{proof}[Proof of \Cref{thm:main}]
Membership in $\PSPACE$ is stated in  \cite[Proposition~1]{LifshitsLohrey2006}.
For hardness, consider the formula~\eqref{eq:source} and use the parameters~\eqref{eq:parameters} and the instruction
SLP~\eqref{eq:program}.

For every instruction $f$, let $Z_f$ be the start variable of
an SLP for $\blk_M(f)$ as constructed in \Cref{sec:codes}. Introduce
additional variables $Y_i$ with the productions
\begin{equation}\label{eq:binary-grammar}
 \begin{aligned}
 Y_{n+1}&\longrightarrow Z_{\Eq},\\
 Y_i&\longrightarrow
 Y_{i+1}\,Z_{\Save_i}\,Z_{\Add_i}\,
 Y_{i+1}\,Z_{\Sub_i}\,Z_{\Ret_i}
 \quad(1\leq i\leq n).
 \end{aligned}
\end{equation}
Let $\ell=5\cdot2^n-4$ (see \eqref{eq:program-length}) and finally produce in polynomial time SLPs for
$u=(0^M1^M)^\ell00$ and $v=0\,\val(Y_1)$.
The SLP for $v$ has polynomial size by \Cref{sec:codes} and
\eqref{eq:binary-grammar}. 
The pattern $u$ has an SLP of size
$\mathcal{O}(\log M+\log\ell)$. 
By \eqref{eq:final-threshold} and~\eqref{eq:root} we have
\[
 u\embed v
 \quad\Longleftrightarrow\quad
 \evalprog{E_1}(1)\geq2
 \quad\Longleftrightarrow\quad
 \Phi\text{ is true}.
\]
This proves $\PSPACE$-hardness. 
\end{proof}

\subsection{An example}\label{sec:example}

\begin{example}
Consider the formula $\forall x_1\in\{0,1\}\ \exists x_2\in\{0,1\}:
 x_1+x_2=1$.
Here $n=2$, $S=2$, $B=16$, $M=48$, $C_1=8$, $C_2=4$, and
$\ell=16$. The sentence is true: the existential choice can be
$x_2=1-x_1$. At the root, $E_2$ returns true for both partial sums
zero and one. The step-by-step computation of $\evalprog{E_1}(1)$ is
\[
 1\xmapsto{E_2}2
  \xmapsto{\Save_1}9
  \xmapsto{\Add_1}25
  \xmapsto{E_2}26
  \xmapsto{\Sub_1}10
  \xmapsto{\Ret_1}2.
\]
For the universal return, the threshold is $C_1+2=10$.
The output pattern is $(0^{48}1^{48})^{16}00$, and it embeds into the
constructed text. If the two quantifiers are reversed, the sentence is
false. Then both universal child calls return false and the step-by-step computation of $\evalprog{E_1}(1)$
is
\[
 1\xmapsto{E_2}1
  \xmapsto{\Save_1}1
  \xmapsto{\Add_1}17
  \xmapsto{E_2}17
  \xmapsto{\Sub_1}1
  \xmapsto{\Ret_1}1.
\]
The same pattern $(0^{48}1^{48})^{16}00$ does not embed into the corresponding text. 
\end{example}

\section{Proof of Theorem~\ref{thm:mainNC1}} \label{sec:uncompressed}

In this section, we prove the $\NC^1$-hardness of the uncompressed subsequence problem $\SubSeq_2$.
As mentioned in \Cref{sec:words}, the greedy embedding algorithm solves this problem in deterministic logspace (also denoted with $\mathsf{L}$), which contains 
$\NC^1$.

We prove \Cref{thm:mainNC1} by reducing the Boolean formula value problem to $\SubSeq_2$.
We consider Boolean formulas over the binary operations $\wedge$ and $\vee$ and the constants $0$ (false) and $1$ (true).
For such a formula $F$, let $\bval(F)\in\{0,1\}$ be its Boolean
value. We use the representation of $F$ by a word over the alphabet $\{ \wedge, \vee, 0, 1, (, ) \}$. The size of $F$ is the length of this word. 
For instance, the word $F = ((1 \wedge 0) \vee ( (0 \wedge 1) \vee 1))$ is the word representation of a
Boolean formula  with $\bval(F)=1$ and size $17$.

Let $T$ be the syntax tree of the Boolean formula $F$, where every node of $T$ is labelled with either $\wedge$, $\vee$, $0$, or $1$.
Clearly, every node of $T$ evaluates to a truth value in the obvious way. The depth of the formula $F$, $\mathsf{depth}(F)$ for short, is
the height of $T$, i.e., the maximal number of edges from the root to a leaf in $T$.

The \emph{Boolean formula value problem} asks whether $\bval(F)=1$ holds for a given Boolean formula $F$. 
Buss has shown that the Boolean formula value problem is complete for $\NC^1$ \cite{Bus87}.
Note that the Boolean negation operator $\neg$ is not needed in \cite{Bus87}.


Our reduction will first balance the input formula, then construct two
straight-line programs of logarithmic depth, and finally expand
their values in $\TC^0$. For the second step we use 
the following results from~\cite{LifshitsLohrey2006} (the following statements are shown in the proofs of 
\cite[Propositions~3 and~4]{LifshitsLohrey2006}).

\begin{lemma}[\mbox{\cite[Prop.~3]{LifshitsLohrey2006}}]\label{lem-and}
Let $u_1,u_2,v_1,v_2\in\{0,1\}^*$ and $N\in\N$ with
$N>\max\{|v_1|,|v_2|\}$. Define $u,v\in\{0,1\}^*$ by
\begin{equation}\label{eq-and-words}
u=u_1 1^N 0 1^N u_2,
\qquad
v=v_1 1^N 0 1^N v_2.
\end{equation}
Then $u\embed v$ if and only if both $u_1\embed v_1$ and
$u_2\embed v_2$.
\end{lemma}

\begin{lemma}[\mbox{\cite[Prop.~4]{LifshitsLohrey2006}}]\label{lem-or}
Let $u_1,u_2,v_1,v_2\in\{0,1\}^*$ and $N\in\N$ with
$N>|u_1|+|u_2|$. Define $u,v\in\{0,1\}^*$ by
\begin{equation}\label{eq-or-words}
u=u_1 0 1^N 0 u_2,
\qquad
v=v_1 0 1^N u_1 0 u_2 1^N 0 v_2.
\end{equation}
Then $u\embed v$ if and only if $u_1\embed v_1$ or
$u_2\embed v_2$.
\end{lemma}

\begin{proof}[Proof of Theorem~\ref{thm:mainNC1}]
Start with a Boolean formula $F$ in word encoding. In the first step, we transform $F$
in $\TC^0$ into (the word encoding of) an equivalent Boolean formula of 
depth $d \leq \mathcal{O}(\log |F|)$; see \cite[Theorem 6.3 and Lemma 3.4]{GanardiL19}. We denote the 
resulting formula again with $F$.
Let $T$ be the syntax tree of $F$.
We next transform $F$ into an SLP $\mathcal G$ by introducing for every node $x$ of $T$ four
variables $U_x$, $V_x$, $M_x$, $N_x$ with the following productions depending on the Boolean operator used at $x$.
\begin{itemize}
\item If $x$ is labelled with $0$ then $U_x \to 0$,  $V_x \to 1$, $M_x \to 1$, $N_x \to 1$
\item If $x$ is labelled with $1$ then $U_x \to 0$,  $V_x \to 0$, $M_x \to 1$, $N_x \to 1$
\item If $x$ is labelled with $\wedge$ and has the children $x_1$ and $x_2$ then 
\begin{eqnarray}
\label{andSLP1} U_x & \to & U_{x_1}  \textcolor{blue}{N_{x_1} N_{x_2}} \, 0 \, \textcolor{blue}{N_{x_1} N_{x_2}}  U_{x_2}, \\
\label{andSLP3} M_x & \to & M_{x_1} \textcolor{blue}{N_{x_1} N_{x_2}} \, 1 \, \textcolor{blue}{N_{x_1} N_{x_2}} M_{x_2}, \\
\label{andSLP2} V_x & \to & V_{x_1}  \textcolor{blue}{N_{x_1} N_{x_2}} \, 0 \, \textcolor{blue}{N_{x_1} N_{x_2}}  V_{x_2}, \\
\label{andSLP4} N_x & \to &  N_{x_1} \textcolor{blue}{N_{x_1} N_{x_2}} \, 1 \, \textcolor{blue}{N_{x_1} N_{x_2}} N_{x_2},
\end{eqnarray}
\item If $x$ is labelled with $\vee$ and has the children $x_1$ and $x_2$ then 
\begin{eqnarray}  \allowdisplaybreaks
\label{orSLP1} U_x &\to & U_{x_1}  \, 0 \, \textcolor{blue}{1 \, M_{x_1} M_{x_2}}  \, 0  \, U_{x_2}, \\
\label{orSLP3} M_x &\to & M_{x_1}  \, 1 \, \textcolor{blue}{1 \, M_{x_1} M_{x_2}}  \, 1  \, M_{x_2}, \\
\label{orSLP2} V_x &\to & V_{x_1}  \, 0 \,  \textcolor{blue}{1\, M_{x_1} M_{x_2}} U_{x_1}  \, 0  \, U_{x_2} \, \textcolor{blue}{1 \,M_{x_1} M_{x_2}}  \,  0  \, V_{x_2} \\
\label{orSLP4} N_x &\to & N_{x_1}  \, 1 \,   \textcolor{blue}{1\, M_{x_1} M_{x_2}} M_{x_1}  \, 1  \, M_{x_2}  \, \textcolor{blue}{1  \, M_{x_1} M_{x_2}}  \,  1  \, N_{x_2}
\end{eqnarray}
\end{itemize}
To understand the above construction, notice that $M_x$ produces exactly $|\val_{\mathcal G}(U_x)|$ many $1$'s and 
$N_x$ produces exactly $|\val_{\mathcal G}(V_x)|$ many $1$'s. This follows directly by induction on the depth of the subtree
rooted in $x$. Also note that in Lemma~\ref{lem-and} one could replace the condition $N>\max\{|v_1|,|v_2|\}$ by
$N\geq |v_1|+|v_2|$ assuming that $v_1 \neq \varepsilon \neq v_2$.  In our SLP $\mathcal{G}$, every variable produces
a non-empty word. Hence, the blue factors in \eqref{andSLP1}--\eqref{orSLP4} produce the $1^N$-factors in \eqref{eq-and-words} and \eqref{eq-or-words}.
The additional $1$'s in the right-hand sides of \eqref{orSLP1}--\eqref{orSLP4} are needed, because Lemma~\ref{lem-or} requires $N\ge 1+|u_1|+|u_2|$.
With this in mind it is clear that
Lemmas~\ref{lem-and} and \ref{lem-or} yield the following equivalence for every node $x$ of $T$:
$\val_{\mathcal G}(U_x) \embed \val_{\mathcal G}(V_x)$ if and only if node $x$ evaluates to true. Hence, if 
$r$ denotes the root of $T$, we have $\val_{\mathcal G}(U_r) \embed \val_{\mathcal G}(V_r)$ if and only if the formula $F$ evaluates to true.

It remains to show that from $\mathcal{G}$ one can produce in $\TC^0$ the words $\val_{\mathcal G}(U_r)$ and $\val_{\mathcal G}(V_r)$.
The crucial point is of course that the depth of $\mathcal{G}$ is $\mathsf{depth}(F)+1$ and hence in $\mathcal{O}(\log |F|)$.
Since every right-hand side of $\mathcal{G}$ has length at most 13, it follows that the lengths
of the words  $\val_{\mathcal G}(U_r)$ and $\val_{\mathcal G}(V_r)$ are polynomially bounded in $|F|$. To see that the construction can be 
carried out in $\TC^0$, we use the framework of \cite{GanardiL19} that we first explain briefly without going too much into the technical details.

Instead of working with $\TC^0$-circuits,
the authors of \cite{GanardiL19} use the concept of FOM-definable functions, where FOM stands for first-order logic with counting.
It is defined as the extension of classical first-order logic with a majority quantifier. It is well-known that a function between finite relational structures
over possibly different signatures is FOM-definable if and only if it is $\TC^0$-computable. For the latter, one has to fix a reasonable binary encoding
of finite relational structures. For our purpose, the technical details for this are not important. We will continue to speak of 
$\TC^0$-computable functions.

The paper \cite{GanardiL19} deals with finite $k$-ordered graphs for $k \geq 1$, which are structures
of the form $(V, E, (P_a)_{a \in \Sigma})$, where $V$ is a finite set of nodes, $E \subseteq \{ (u,i,v) : u,v \in V, i \in [0,k-1] \}$ is a set of $[0,k-1]$-labelled
edges and $P_a \subseteq V$ for $a \in \Sigma$ is the set of $a$-labelled nodes. In addition, it is required that for every $u \in V$ and $i \in [0,k-1]$
there is at most one node $v \in V$ with $(u,i,v) \in E$.
Both syntax trees of Boolean formulas and the SLPs that arise in the above construction (more precisely, we take the DAG of an SLP, as defined in Section~\ref{sec:slp},
which we identify with the SLP in the following), fit into this framework. For a syntax tree we take $k=2$ and $\Sigma = \{0,1,\wedge,\vee\}$, and for the SLP constructed
from $T$ as above, we take $k = 13$ (the maximal length of a right-hand side 
in \eqref{andSLP1}--\eqref{orSLP4}) and $\Sigma = \{0,1\}$ (the terminal alphabet).

Another important concept in \cite{GanardiL19} is the  so-called extended connection representation of a $k$-ordered graph $G = (V,E,(P_a)_{a \in \Sigma})$.
It is obtained by replacing the edge relation $E \subseteq \{ (u,i,v) : u,v \in V, i \in [0,k-1] \}$ by the relation 
$E_{\mathsf{ec}}$ that consists of all tuples $(u,\alpha,v)$ with $u,v \in V$, $\alpha \in [0,k-1]^*$ such that $|\alpha| \leq \log_k |V|$ 
and $u \cdot \alpha = v$. Here, $u \cdot \alpha$ is the node obtained by navigating from $u$ according to the directions in $\alpha$.
Formally, it is defined inductively: $u \cdot \varepsilon = u$ and, if $i \in [0,k-1]$ and $\alpha \in [0,k-1]^*$, then 
$u \cdot (i \alpha) = v \cdot \alpha$ where $(u,i,v) \in E$. If there is no such edge $(u,i,v)$, then 
$u \cdot (i \alpha)$ is undefined, otherwise there is exactly one such edge $(u,i,v)$.
Let us write $\mathsf{ec}(G)$ for the  extended connection representation of $G$.

Lemmas~3.1 and~3.2 from \cite{GanardiL19} together imply that from a well-parenthesized expression $E$ that is given by a word with brackets
one can compute $\mathsf{ec}(T)$, for $T$ the syntax tree of $E$, in $\TC^0$.
 This result can be directly used to compute from the word encoding of the Boolean formula $F$ the structure $\mathsf{ec}(T)$, where
 $T$ is the syntax tree of $F$. Lemma~3.1 from \cite{GanardiL19} first adds the ancestor relation to the syntax tree $T$. This enriched structure is then transformed
 into the extended connection representation of $T$ using \cite[Lemma~3.2]{GanardiL19}.
 
The crucial step is then to compute  $\mathsf{ec}(\mathcal{G})$ in $\TC^0$ from $\mathsf{ec}(T)$, where
$\mathcal{G}$ is the SLP that we have constructed above.
For this, it is more convenient to modify the SLP $\mathcal{G}$ and the corresponding DAG slightly without changing the derived words.
Recall that in the DAG defined in Section~\ref{sec:slp} there is for every terminal symbol $a$ a unique 
node in the DAG that is labeled with $a$. In our modified SLP we introduce for every inner node $x$ of the syntax tree $T$ two more 
nonterminals $Z_x$ and $O_x$ (for zero and one) and replace the $0$'s (resp., $1$'s) in the right-hand sides of 
 \eqref{andSLP1}--\eqref{orSLP4} by $Z_x$ (resp., $O_x$). The new nonterminals $Z_x$ and $O_x$ become nodes of the DAG, where
 they are labelled with $0$ and $1$, respectively, and, importantly, they have no outgoing edges in the DAG.
 Analogously all old nonterminals $X$ of $\mathcal{G}$ with a production $X \to a$ for $a \in \{0,1\}$ have outdegree zero 
 and are labelled with $a$ in the new DAG that we call $\mathcal{D}$. 
 Note that $\mathcal{D}$ has no extra nodes for the two terminals $0$ and $1$.
 
 It is easy to see that $\mathcal{D}$ can be computed in $\TC^0$ from the syntax tree $T$. This is in fact a purely local computation
 (or, in the language of \cite{GanardiL19}, it is FOM-definable).
 For every inner node $x$ of $T$, we have six nodes in the DAG ($U_x$, $V_x$, $M_x$, $N_x$, $Z_x$, $O_x$), that can be encoded by the pairs 
 $(x,U)$, $(x,V)$, $(x,M)$, $(x,N)$, $(x,Z)$, and $(x,O)$.
 For a leaf $x$ of $T$ we only have the four DAG nodes  $(x,U)$, $(x,V)$, $(x,M)$, $(x,N)$. Whether 
 $((x,X), i, (y,Y))$ is an edge in the DAG
 only depends on (i) $i$, $X$ and $Y$, (ii) the node label of $x$ and (iii) the Boolean values of the statements $x=y$ and $(x,j,y) \in E$ for $j \in [0,k-1]$ and $E$ the edge relation of $T$.
  A similar remark applies  to the node labels in the DAG.
 
 In order to lift the $\TC^0$-computation of $\mathcal{D}$ from $T$ to a 
 $\TC^0$-computation of $\mathsf{ec}(\mathcal{D})$ from $\mathsf{ec}(T)$, we need the fact that the transformation from
 $T$ to $\mathcal{D}$ is a so-called guarded graph transduction
in the sense of \cite{GanardiL19}. We need only the one-dimensional case (corresponding to $m=1$ in \cite{GanardiL19}).
Then guardedness boils down to the property that for every edge $((x,X), i, (y,Y))$ of $\mathcal{D}$ 
we have $x=y$ or $(x,j,y)$ is an edge of $T$ for some $j \in \{1,2\}$.
Lemma 3.7 from \cite{GanardiL19} says that if a transformation $G \mapsto f(G)$ on $k$-ordered graphs ($k$ can be different for $G$ and $f(G)$) 
is a $\TC^0$-computable guarded graph transduction then also the transformation $\mathsf{ec}(G) \mapsto \mathsf{ec}(f(G))$ is
 $\TC^0$-computable. In our setting, this implies that $\mathsf{ec}(\mathcal{D})$ can be computed in $\TC^0$ from $\mathsf{ec}(T)$.
 

Since we now have the extended connection representation of $\mathcal{D}$ available one can use
 \cite[Lemma~3.4]{GanardiL19} to unfold $\mathcal{D}$ in $\TC^0$. More precisely, we can compute in $\TC^0$
from $\mathsf{ec}(\mathcal{D})$ the well-parenthesized word representations of the syntax trees $T_{U_r}$
and $T_{V_r}$  that are obtained by unfolding the DAG $\mathcal{D}$ starting from the variables $U_r$, respectively $V_r$.
For this, the logarithmic depth of the SLP $\mathcal{G}$ is important.
Finally, we can obtain the words $\val_{\mathcal G}(U_r)$ and $\val_{\mathcal G}(V_r)$ from the word representation of the syntax trees $T_{U_r}$
and $T_{V_r}$ by removing from the word representations all symbols except the terminal symbols (that means, we remove all brackets). This is a morphism between free monoids, which can be computed also in $\TC^0$ \cite{LangeMcKenzie1998}.
\end{proof}

\subparagraph{Disclosure of AI use.}
A first proof of \Cref{thm:main} was obtained with the help of GPT-6 Astra. The author simplified the proof and polished the presentation. Apart from that
GPT-6 Astra was used for improving the final presentation and eliminating typos.
The author takes full responsibility for the final content of the paper.


\end{document}